\documentclass[a4paper,11pt,oneside,reqno]{amsproc}

\usepackage[colorlinks]{hyperref}
\usepackage{mathrsfs,amsthm}
\usepackage{float}

\newtheorem{theorem}{\rm\bf Theorem}[section]
\newtheorem{proposition}[theorem]{\rm\bf Proposition}

\newtheorem*{theorem*}{Theorem}
\newtheorem*{theorem 1}{\rm\bf Proposition 1}
\newtheorem*{theorem 2}{\rm\bf Proposition 2}
\newtheorem*{conj*}{Conjecture}

\theoremstyle{definition}

\theoremstyle{remark}
\newtheorem{remark}[theorem]{\rm\bf Remark}

\begin{document}
	\author	{S. Ghersheen}
	\author{V. Kozlov}
	\author{U. Wennergren}
	\title [Specifics of coinfection and it's dynamics] {Specifics of coinfection and it's dynamics}


	\maketitle

		\begin{abstract}
				It is essential to understand the dynamics of epidemics in the presence of coexisting pathogens. There are various phenomenon that can effect the dynamics. In this paper, we formulate a mathematical model using different assumptions  to capture the effect of these additional phenomena such as partial cross immunity, density dependence in each class and a role of recovered population in the dynamics. We found the basic reproduction number for each model which is the threshold that describes the invasion of disease in population. The basic reproduction number in each model shows that the persistence of disease or strains depends on the carrying capacity. In the model of this paper, we present the local stability analysis of the  boundary equilibrium points and observed that the recovered population is not uniformly bounded with respect to $K$.
	
	\end{abstract}	

		\section{Introduction}
		The evolution and epidemiology of parasite virulence has been the main research focus from many years. The aim is to study  and understand different traits like persistence, interaction, role of immunity and treatment of diseases. Each of this trait has its own importance and with the time, due to the (emergence and evolution) of new pathogens, is becoming more complicated to understand. The most complicated issue is understanding how parasite diversity within individual hosts alters the selective forces determining virulence. We are interested to analyse the dynamics of when the host is infected by the different specie of parasite or with the number of different parasite species at a time. This phenomenon commonly known as coinfection.
		Coinfection of different pathogen species or different pathogen has not been completely understandable yet due to its complexities. In order to provide the deeper understanding of  evolution and dynamics of multiple pathogens virulence researchers are focused now on the interaction of these parasite species and the consequences of  their emergence in the host. We have analysed, in \cite{SKTW18a}, \cite{SKTW19}, \cite{SKTW21a} and \cite{SKTW21b},  models for coinfection under different assumptions on the interaction between strains to understand the role of crucial parameters in the dynamics. Since Multi-strain pathogens impose a substantial burden of morbidity and mortality \cite{Gupta}.There are many studies \cite{Allen,Martcheva,Blyuss} which explains the severity of coinfection in general. Also in \cite{Castillo2,Kawaguchi,Alemu,Chaves} several models of disease specific coinfection has been studies.Some ting more
		 In this paper, the aim is to address the some possible situations which are common in the case of single infection but are not been studied in detail yet in the case of  coinfection.e discuss the effect of density dependence in each class and influence of immunity and partial cross immunity on disease dynamics in coinfection case.
		In the first model we consider that each class density is restricted by carrying capacity $K$ . This model is an extension of our previous models in \cite{SKTW18a,SKTW19} and involves many complexities. However to make it analytically solvable, we have certain assumption in this case. We assume that the recovered population is completely immune from disease or strains upon the recovery and also the new born in recovered class have life long  immunity.
		In the next model we consider that the new born in the recovered class are not immune. They born as a susceptible and become the part of susceptible population which is realistic in the case of many diseases.
		One of the crucial features in multi pathogen coexistence is the presence of cross immunity, whereby infection by one strain induces partial immunity or complete cross immunity towards future. We discuss the case of complete cross immunity in the first model and in model 4, we relax this assumption and consider the case of partial cross immunity. This bring more complexity in the dynamics.
		In all our previous and present models we consider the mass action type transmission to diminish further complexity. It is interesting to see the role bilinear  incidence rates in disease dynamics which has been considered in one of the models given in this paper.
		We observed that the  dynamics of disease depends on carrying capacity in each model.

	\section{Formulation of the models}\label{sec:model}

In this section we formulate a model including coinfection of two infectious agents in population. The model given below is an SIR model where we consider density dependence in each class.
	
		We formulate an SIR model with the recovery of each class and assume that infected and recovered populations can reproduce. We also assume that a susceptible individual can be infected with both strains or any of the single strain as a result of contact with coinfected person. Moreover coinfection can occur as a result of contact between two single infection individuals or coinfected and single infected individuals. However, to make model analytically solvable, we assume that the recovered class is completely immune after recovery and we consider pseudo-mass action incident rates or bilinear incident rate and model then is given by the following system of differential equations
	\begin{equation}\label{Model}
	\begin{split}
	S' &=\biggl(b(1-\frac{N}{K})-\alpha_1I_1-\alpha_2I_2-(\beta_1+\beta_2+\alpha_3)I_{12}-\mu_0\biggr)S,\\
	I_1' &=\left(b(1-\frac{N}{K})+\alpha_1S - \eta_1I_{12}-\gamma_1I_2 - \mu_1\right)I_1+\beta_1SI_{12},\\
	I_2' &=\left(b(1-\frac{N}{K})+\alpha_2S - \eta_2I_{12}-\gamma_2I_1- \mu_2\right)I_2+\beta_2SI_{12},\\
	I_{12}' &=\left(b(1-\frac{N}{K})+\alpha_3S+ \eta_1I_1+\eta_2I_2-\mu_3\right)I_{12}+(\gamma_1+\gamma_2)I_1I_2, \\
	R' &=\left(b(1-\frac{N}{K})-\mu_4'\right)R+\rho_1 I_1+\rho_2I_2+\rho_3 I_{12},
	\end{split}
	\end{equation}
	where $S$ represents the susceptible class, $I_1$ and $I_2$ are infected classes from strain 1 and strain 2 respectively, $I_{12} $ represents co-infected class, $R$ represents the recovered class. Also,
	$$N=S+I_1+I_2+I_{12}+R
	$$
	is the total population. We use the following parameters
	\begin{itemize}
		\item
		$b$ is the birthrate of each class  
		\item
		$K$ is carrying capacity of class ;
		\item
		$\rho_i$ is recovery rate from each infected class ($i=1,2,3$);
		\item
		$\beta_i$ is the rate of transmission of single infection from coinfected class ($i=1,2$);
		\item
		$\gamma_i$ is the rate at which infected with one strain get infected with  the other strain and move to coinfected class ($i=1,2$);
		\item
		$\mu_i' $ is death rate of each class,  $( i=1,2,3,4)$ and $\mu_i=\rho_i+\mu_i', i=1,2,3$;
		\item
		$\alpha_1$, $\alpha_2$, $\alpha_3$  are rates of transmission of strain 1, strain 2 and both strains (in the case of coinfection),
		\item
		$\eta_i$ is rate at which infected from one strain getting  infection from co-infected class $( i=1,2).$
	\end{itemize}

The possibility of simultaneous transmission from a single contact with a dually infected individual, which we model according to the assumption $\alpha_3 > 0$. An underlying assumption in the model is that individuals in all disease states have the same contact rate; we do not assume that individuals have fewer contacts if they are infected. Moreover, it is reasonable to assume that  birth rate $b>\mu_j,~j=1,2,3,4,$ otherwise population dies out quickly.
 We also assume here the case when the death rate of recovered population is greater than the death rate of susceptible population, i.e

 \begin{equation}\label{assum2}
 \mu_0< \mu_4'<\mu_j', j=1,2,3.
 \end{equation}
Since, in this model we assume the complete cross immunity therefore, the death rate of recovered class is relatively lower than other infected classes. It is also reasonable to assume that the death rates due to coinfection are greater than the death rate due to single infection due to the severity of coinfection.

 Moreover,  we assume that
 \begin{equation}
 \sigma_1< \sigma_2< \sigma_3,
 \end{equation}
 where
 $$ \sigma_k=\frac{\mu_k-\mu_0}{\alpha_k},~k=1,2,3$$

 Another important parameter for our study is the modified carrying capacity defined by
 $$S^{**}=\frac{K}{b}(b-\mu_0).$$
 which directly dependent on $K$ and a key parameter in dynamics.
 \begin{remark}
 	Note that we are not considering here the case when $$
 	\mu_4'< \mu_0<\mu_j', j=1,2,3.
 	$$
 	Certainly it is also important to consider that case but since, course of infection can make the individual's immune system weak therefore it is possible that their death rate is higher than the susceptible ones. We leave the more detail analysis of this case to the future work.
 \end{remark}

 \section{Boundedness of solutions}\label{5}
 In this section first, we show that   the solution of system \eqref{Model} is  bounded and next, we show some important estimates for  boundedness of total population $N.$
 \begin{proposition}
 	If $(S,I_1,I_2,I_{12},R)(t)$ is a solution of \eqref{Model} with $S(0)$ positive then
 	\begin{equation}\label{colo}
 	S(t) \leq \frac{1}{\frac{1}{S^{**}}(1-e^{-(b-\mu_0)t})+\frac{1}{S(0)}e^{-(b-\mu_0)t}} .
 	\end{equation}
 	In particular,
 	\begin{equation}\label{boundr1}
 	S(t)\leq \max\{ S^{**},S(0)\}
 	\end{equation}
 	and
 	\begin{equation}\label{limit}
 	\limsup_{t\rightarrow \infty} S(t)\leq  S^{**}.
 	\end{equation}
 \end{proposition}
 \begin{proof}
 	It follows from the first equation of \eqref{Model}
 	\begin{equation*}
 	S'-(b-\mu_0)S \leq-\frac{bS^2}{K},
 	\end{equation*}
 	which can be written as
 	\begin{equation*}
 	(Se^{-(b-\mu_0)t})' \leq - \frac{b}{K} e^{-(b-\mu_0)t}S^2.
 	\end{equation*}
 	
 	Dividing both sides by $(Se^{-(b-\mu_0)t})^2$ and integrating from $0$ to $t$ gives,
 	\begin{equation*}
 	\frac{e^{(b-\mu_0)t}}{S} \geq \frac{b}{K(b-\mu_0)} (e^{(b-\mu_0)t}-1) + \frac{1}{S(0)},
 	\end{equation*}
 	which leads to \eqref{colo}.
 	
 	Relations \eqref{boundr1} and \eqref{limit} are direct consequences of \eqref{colo}.
 \end{proof}

 \begin{proposition}
 	The total population $N(t)$ satisfies
 \begin{enumerate}
 			\item[(i)]
 		 \begin{equation}\label{colo1}
 	N(t) \leq \frac{1}{\frac{1}{S^{**}}(1-e^{-(b-\mu_0)t})+\frac{1}{N(0)}e^{-(b-\mu_0)t}} .
 	\end{equation}
 	In particular,
 	\begin{equation}\label{boundr2.1}
 	N(t)\leq \max\{ S^{**},N(0)\}\quad\text{and}\quad\limsup_{t\rightarrow \infty} N(t)\leq  S^{**}.
 	\end{equation}

 	
 	\item[(ii)]
 	\begin{equation}\label{colo3.1}
 	N(t) \geq \frac{1}{\frac{b}{K(b-\mu_3')}(1-e^{-(b-\mu_3')t})+\frac{1}{N(0)}e^{-(b-\mu_3')t}} .
 	\end{equation}
 	In particular,
 	\begin{equation}\label{boundr3.1}
 	N(t)\geq \min\{ \frac{K}{b}(b-\mu_3'),N(0)\}\quad\text{and}\quad\liminf_{t\rightarrow \infty} N(t)\geq  \frac{K}{b}(b-\mu_3').
 	\end{equation}
 	\end{enumerate}
 \end{proposition}
 \begin{proof}
 	\begin{enumerate}
 			\item[(i)]
 	Summing up all the equations of \eqref{Model} gives
 	\begin{equation*}
 	N'-(b-\mu_0)N \leq-\frac{bN^2}{K},
 	\end{equation*}
 	which can be written as
 	\begin{equation*}
 	(Ne^{-(b-\mu_0)t})' \leq - \frac{b}{K} e^{-(b-\mu_0)t}N^2.
 	\end{equation*}
 	
 	Dividing both sides by $(Ne^{-(b-\mu_0)t})^2$ and integrating from $0$ to $t$ gives,
 	\begin{equation*}
 	\frac{e^{(b-\mu_0)t}}{N} \geq \frac{1}{S^{**}} (e^{(b-\mu_0)t}-1) + \frac{1}{N(0)},
 	\end{equation*}
 	which leads to \eqref{colo1}.
 	
 	Relations in \eqref{boundr2.1} are direct consequences of \eqref{colo1}.

 		\item[(ii)]
 	Summing up all the equations of \eqref{Model} gives
 	\begin{equation*}
 	N'-(b-\mu_3')N \geq-\frac{bN^2}{K},
 	\end{equation*}
 	which can be written as
 	\begin{equation*}
 	(Ne^{-(b-\mu_3')t})' \geq - \frac{b}{K} e^{-(b-\mu_3')t}N^2.
 	\end{equation*}
 	
 	Dividing both sides by $(Ne^{-(b-\mu_3')t})^2$ and integrating from $0$ to $t$ gives,
 	\begin{equation*}
 	\frac{e^{(b-\mu_3')t}}{N} \leq \frac{b}{K(b-\mu_3')} (e^{(b-\mu_3')t}-1) + \frac{1}{N(0)},
 	\end{equation*}
 	which leads to \eqref{colo3.1}.
 	
 	Relations in \eqref{boundr3.1} are direct consequences of \eqref{colo3.1}.
 \end{enumerate}
 \end{proof}
 \subsubsection{Boundaries for $N$}
 Here we discuss the boundaries for an equilibrium state.

 Summing up all the equations in \eqref{Model} at $S,I_1,I_2,I_{12},R$ non-zero gives
 \begin{equation}
 b(1-\frac{N}{K})N-\mu_0S-\mu_1'I_1-\mu_2'I_2-\mu_3'I_{12}-\mu_4'R=0.
 \end{equation}
 Above equation can be written as
 \begin{equation}\label{Neq}
 \frac{b}{K}(N^{**}-N)N-(\mu_1'-\mu_0)I_1-(\mu_2'-\mu_0)I_2-(\mu_3'-\mu_0)I_{12}-(\mu_4'-\mu_0)R=0,
 \end{equation}
 where $N^{**}=K(1-\frac{\mu_0}{b}).$

 It follows from \eqref{Neq} and \eqref{assum2}
 $$N\leq N^{**}.$$
 Moreover, it follows from the first equation of \eqref{Model}
 \begin{equation}\label{Seq}
 \frac{b}{K}(N^{**}-N)-\alpha_1I_1-\alpha_2I_2-\hat\alpha_3I_{12}=0.
 \end{equation}
 Dividing equations \eqref{Neq} by \eqref{Seq} gives
 \begin{equation}
 N=\frac{(\mu_1'-\mu_0)I_1+(\mu_2'-\mu_0)I_2+(\mu_3'-\mu_0)'I_{12}+(\mu_4'-\mu_0)R}{\alpha_1I_1+\alpha_2I_2+\hat\alpha_3I_{12}}.
 \end{equation}

 If $I_1+I_2+I_{12}\neq 0$ then

 \begin{equation}
 \begin{split}
 N&\geq \min_{I_1+I_2+I_{12}\neq 0} \frac{(\mu_1'-\mu_0)I_1+(\mu_2'-\mu_0)I_2+(\mu_3'-\mu_0)'I_{12}}{\alpha_1I_1+\alpha_2I_2+\hat\alpha_3I_{12}},\\
 &=\min\left(\frac{\mu_1'-\mu_0}{\alpha_1},\frac{\mu_2'-\mu_0}{\alpha_2},\frac{\mu_3'-\mu_0}{\alpha_3}\right)=\sigma_1.
 \end{split}
 \end{equation}
 Hence $\sigma_1\leq N\leq N^{**}.$
 \subsubsection{Boundaries of $N-R$}
 Summing up  the  first four equations in \eqref{Model} at an equilibrium state $S,I_1,I_2,I_{12}$ gives
 \begin{equation}
 b(1-\frac{N}{K})(N-R)-\mu_0S-\mu_1I_1-\mu_2I_2-\mu_3I_{12}=0.
 \end{equation}
 Above equation can be written as
 \begin{equation}\label{Neq1.1}
 \frac{b}{K}(S^{**}-N)(N-R)-(\mu_1-\mu_0)I_1-(\mu_2-\mu_0)I_2-(\mu_3-\mu_0)I_{12}=0.
 \end{equation}
 Moreover, it follows from the first equation of \eqref{Model}
 \begin{equation}\label{Seq1.2}
 \frac{b}{K}(S^{**}-N)-\alpha_1I_1-\alpha_2I_2-\hat\alpha_3I_{12}=0.
 \end{equation}
 Dividing equations \eqref{Neq1.1} by \eqref{Seq1.2} gives
 \begin{equation}
 N-R=\frac{(\mu_1-\mu_0)I_1+(\mu_2-\mu_0)I_2+(\mu_3-\mu_0)'I_{12}}{\alpha_1I_1+\alpha_2I_2+\hat\alpha_3I_{12}}.
 \end{equation}

 If $I_1+I_2+I_{12}\neq 0$ then $
 \sigma_1\leq N-R\leq \sigma_3.$

 \section{ Equilibrium points and their properties}\label{eqpoint}
In this section we discuss the all equilibrium points and some of their  basic properties.  We rewrite the system \eqref{Model} for equilibrium points as follows
 	\begin{equation}\label{EqptModel}
 \begin{split}
\biggl(b\left(1-\frac{N}{K}\right)-\alpha_1I_1-\alpha_2I_2-(\beta_1+\beta_2+\alpha_3)I_{12}-\mu_0\biggr)S&=0,\\
\left(b(1-\frac{N}{K})+\alpha_1S - \eta_1I_{12}-\gamma_1I_2 - \mu_1\right)I_1+\beta_1SI_{12}&=0,\\
 \left(b(1-\frac{N}{K})+\alpha_2S - \eta_2I_{12}-\gamma_2I_1- \mu_2\right)I_2+\beta_2SI_{12}&=0,\\
 \left(b(1-\frac{N}{K})+\alpha_3S+ \eta_1I_1+\eta_2I_2-\mu_3\right)I_{12}+(\gamma_1+\gamma_2)I_1I_2&=0,\\
\left(b(1-\frac{N}{K})-\mu_4'\right)R+\rho_1 I_1+\rho_2I_2+\rho_3 I_{12}&=0.
 \end{split}
 \end{equation}
 In the next subsection, we discuss those equilibrium points which do not exist or are always unstable.
   \subsection{Equilibrium points with  one zero $ I_i$-component and all zero $ I_i$-components }
 There are two equilibrium point where one of the  $I$ component is zero and coinfected class is present and one equilibrium point where both single infections are non zero but coinfected class is zero. There is one equilibrium point where both single infections are absent but the component $I_{12}$ is non zero. These equilibrium points do not exist because it directly follows from the second and third equation of \eqref{EqptModel} if $I_1=I_2=0$ then $I_{12}$ must be zero. It means that the coinfection can not exist  in the absence of both single infectious agents.
 Next we consider the equilibrium point with $I_1=I_2=I_{12}=0$ but $S\neq 0$ and $R\neq0.$ i.e $$
 g_0=(S^*,0,0,0,R^*).
 $$
 Then it follows from the first and last equation of \eqref{EqptModel} that
 $$b\left(1-\frac{N}{K}\right)=\mu_0\quad\text{and}\quad b\left(1-\frac{N}{K}\right)=\mu_4',$$
 which have no common solution since  $\mu_0\neq \mu_4'.$
  \subsection{Equilibrium points with zero S or R-component }
  Let $R=0$ then  it follows from the last equation of \eqref{EqptModel} that $I_1=I_2=I_{12}=0$ and there are only two such equilibrium points in this case $G_1=(0,0,0,0,0)$ and $G_2=(S^{**},0,0,0,0). $ The equilibrium point $G_1$ is always unstable.
  Next we consider the case when $R\neq 0$ and $S=0. $

The first equilibrium point in this category is $g_1=(0,0,0,0,R^{**})$ and this point is unstable.
 The next three  equilibrium points in this case are
 \begin{equation}\label{eqpt8}
 g_2=\left(0,I_1^*,0,0,R^*\right),
 \end{equation}
 where $$ I_1^*=\frac{K}{b}(b-\mu_1)-R^*,\quad R^*=\frac{\rho_1(b-\mu_1)}{(\mu_4'-\mu_1')}.$$

 \begin{equation}\label{eqpt9}
 g_3  =\left(0,0,I_2^*,0,R^*\right),
 \end{equation}
 where $$ I_2^*=\frac{K}{b}(b-\mu_2)-R^*,\quad R^*=\frac{\rho_2(b-\mu_2)}{(\mu_4'-\mu_2')}$$
 and
 \begin{equation}\label{eqpt10}
 g_4 =\left(0,0,0,I_{12}^*,R^*\right),
 \end{equation}
 where $$ I_{12}^*=\frac{K}{b}(b-\mu_3)-R^*\quad R^*=\frac{\rho_3(b-\mu_3)}{(\mu_4'-\mu_3')}.$$
 By \eqref{assum2}, the $R$-component in the equilibrium points $g_2, g_3, g_4$ is negative .Therefore these equilibrium points do not exist.
 Let us now consider the case when one of the infected class is zero with $S=0.$
 The next three equilibrium point in this categories are  as follows

 \begin{equation}
 g_4 =\left(0,I_1^*,0,I_{12}^*,R^*\right),
 \end{equation}
 where $I_1^*,I_{12}^*,R^*$ can be found by solving the following equations
 \begin{equation}\label{eqpt11}
 \begin{split}
 b\left(1-\frac{N}{K}\right)-\eta_1I_{12}-\mu_1&=0,\\
 b\left(1-\frac{N}{K}\right)+\eta_1I_1-\mu_3&=0,\\
 \left(b\left(1-\frac{N}{K}\right)-\mu_4'\right)R+\rho_1I_1+\rho_3I_{12}&=0.
 \end{split}
 \end{equation}
 It follows from the \eqref{eqpt11}, that
 \begin{equation}
 \begin{split}
 b\left(1-\frac{N}{K}\right)-\mu_1&>0,\\
 b\left(1-\frac{N}{K}\right)-\mu_4'<0.
 \end{split}
 \end{equation}
 This implies that
 $\mu_1<b\left(1-\frac{N}{K}\right)<\mu_4',$
 which contradicts \eqref{assum2}. Therefore this equilibrium point does not exist.
 \begin{equation}
 g_5=\left(0,0,I_2^*,I_{12}^*,R^*\right).
 \end{equation}
 where $I_2^*,I_{12}^*,R^*$ can be found by solving the following equations
 \begin{equation}\label{eqpt12}
 \begin{split}
 b\left(1-\frac{N}{K}\right)-\eta_2I_{12}-\mu_2&=0,\\
 b\left(1-\frac{N}{K}\right)+\eta_2I_2-\mu_3&=0,\\
 \left(b\left(1-\frac{N}{K}\right)-\mu_4'\right)R+\rho_2I_2+\rho_3I_{12}&=0.
 \end{split}
 \end{equation}
 It follows from the \eqref{eqpt12}, that
 \begin{equation}
 b\left(1-\frac{N}{K}\right)-\mu_2>0\quad \text{and}\quad b\left(1-\frac{N}{K}\right)-\mu_4'<0.
 \end{equation}
 Therefore
 $\mu_2<b\left(1-\frac{N}{K}\right)<\mu_4',$
 which contradicts \eqref{assum2}. Hence this equilibrium point does not exist.
 Similarly  for equilibrium point
 $
  g_6=(0,I_1^*,I_2^*,I_{12}^*,R^*),
 $ we have
 $$\mu_1<\mu_2<b\left(1-\frac{N}{K}\right)<\mu_4, $$ which contradicts \eqref{assum2}. Therefore  $g_6$ does not exist.

 \subsection{Stable Equilibrium points}
 The system \eqref{Model} has thirty two equilibrium points and the equilibrium points with zero $S,R$-components do not exist. Moreover those equilibrium points in which coinfection is present with one single infection and without any single infection do not exist. This guarantees that the coinfection  is mediated by two single infectious agents.  Our aim is to describe here stable equilibrium points.  Note that all equilibrium points with zero $S$ or $R$- components do not exist. 
Also the equilibrium points with one non-zero $I$ -component do not exist. Therefore there are seven non-zero equilibrium points in this case.
 There is one stable non-trivial  disease-free  (or a healthy equilibrium) equilibrium point  which means disease is absent. In the present notation the disease-free equilibrium corresponds to $I_1=I_2=I_{12}=R=0$, and it follows from \eqref{EqptModel}  that

 $$
 G_2=(S^{**},0,0,0,0).
 $$



 The next equilibrium points must have non zero $S$ and $R$ components.
 The first equilibrium point with the presence of the first strain is
 \begin{equation*}
 G_3 =\left(S^*,I_1^*, 0,0,R^*\right).
 \end{equation*}
 where $S^*,I_1^*,R^*$ can be found by solving the following equations
  \begin{equation}\label{eqpt4}
 \begin{split}
 b\left(1-\frac{N}{K}\right)-\alpha_1I_1-\mu_0=0,\\
 b\left(1-\frac{N}{K}\right)+\alpha_1S-\mu_1=0,\\
 b\left(1-\frac{N}{K}\right)-\mu_4'+\rho_1\frac{I_1}{R}=0
 \end{split}
 \end{equation}
and  equilibrium point with the presence of the second strain is
 \begin{equation*}
 G_4 =\left(S^*,0,I_2^*,0,R^*\right),
 \end{equation*}
 where $S^*,I_2^*,R^*$ can be found by solving the following equations
  \begin{equation}\label{eqpt5}
 \begin{split}
 b\left(1-\frac{N}{K}\right)-\alpha_2I_2-\mu_0=0,\\
 b\left(1-\frac{N}{K}\right)+\alpha_2S-\mu_2=0,\\
 b\left(1-\frac{N}{K}\right)-\mu_4'+\rho_2\frac{I_2}{R}=0.
 \end{split}
 \end{equation}
 Note that  equilibrium points $G_3, G_4$ are independent of parameters $\beta_i, \gamma_i,~i=1,2.$

The  remaining equilibrium point is coexistence equilibrium point i.e
$$
G_5=(S^*,I_1^*,I_2^*,I_{12}^*,R^*)
$$
with all non zero components.

\bigskip
\section{Equilibrium point: Existence and Local Stability}
In this section we discuss the existence  and local stability of equilibrium points $G_2,G_3, G_4$ and $G_5.$
\subsection{Equilibrium point $G_2$}
The equilibrium point $G_2$ always exists.
The Jacobian matrix evaluated at $G_2=(S^{**},0,0,0,0)$ is
\[
J=
\begin{bmatrix}
-(b-\mu_0) & -(\alpha_1+\frac{b}{K})S^{**} & -(\alpha_2+\frac{b}{K})S^{**} & -(\hat{\alpha_3}+\frac{b}{K})S^{**}& -(b-\mu_0)  \\
0 & \alpha_1(S^{**}-\sigma_1) & 0 & \beta_1S^{**} & 0 \\
0 & 0 & \alpha_2(S^{**}-\sigma_2) & \beta_2S^{**} & 0 \\
0 & 0 & 0 & \alpha_3(S^{**}-\sigma_3) & 0\\
0 & \rho_1 & \rho_2 & \rho_3 & \mu_0-\mu_4'
\end{bmatrix}.
\]

The second, third and fourth diagonal  elements are negative if
$S^{**}<\sigma_1$
and by \eqref{assum2} the remaining diagonal elements are also negative.
Therefore  the equilibrium point $G_2$ is stable if $S^{**}<\sigma_1$ or equivalently  $K<K^*$
where
\begin{equation}\label{Kcrit}
K^*=\frac{b\sigma_1}{b-\mu_0}.
\end{equation}

\subsection{Equilibrium point $G_3$}
\subsubsection{Existence}
It follows from the first and second equation of \eqref{eqpt4} that
\begin{equation}\label{sum3.1}
S+I_1=\frac{\mu_1-\mu_0}{\alpha_1}=\sigma_1.
\end{equation}
Also from the  first  equation of  \eqref{eqpt4} we get
\begin{equation}\label{G3.11}
I_1=\frac{b}{K\alpha_1}(S^{**}-\sigma_1-R)=\frac{b}{K\alpha_1}(\hat S_1-R),\quad \hat S_1=S^{**}-\sigma_1.
\end{equation}
The components $R$ and $I_1$ are positive if $\hat S_1>0$ and $0<R<\hat S_1.$
We obtain equation for $R$ by substituting $I_1$ and $S=\sigma_1-I_1$ in the last equation of \eqref{eqpt4}  and dividing by $\frac{b}{K}$  gives
\begin{equation}\label{G3.12}
f(R,K)=0,
\end{equation}
where
$$f(R,K)=\left((\hat S_1-R)+\frac{K}{b}(\mu_0-\mu_4')\right)R+\frac{\rho_1}{\alpha_1}(\hat S_1-R).$$
In what follows, it is convenient to consider the component of the equilibrium point as a function of $K.$
Next differentiating the \eqref{G3.12} w.r.t $R$ we obtain that
\begin{equation}\label{G3.12.1}
\frac{\partial f}{\partial R}=\left((\hat S_1-R)+\frac{K}{b}(\mu_0-\mu_4')\right)-R-\frac{\rho_1}{\alpha_1}
\end{equation}
and it follows from \eqref{G3.12}  that $\left((\hat S_1-R)+\frac{K}{b}(\mu_0-\mu_4')\right)<0.$ Therefore $\frac{\partial f(R,K)}{\partial R}<0.$

Moreover
\begin{equation}\label{G3.13}
\frac{\partial f(R,K)}{\partial K}=\left(\frac
{b-\mu_0} {b}+\frac
{\mu_0-\mu_4'} {b}\right)R+\frac{\rho_1}{\alpha_1}(\frac
{b-\mu_0} {b}).
\end{equation}
Since
$$\frac{df(R,K)}{dK}=\frac{\partial f}{\partial R}\frac{dR}{dK}+\frac{\partial f}{\partial K}=0,$$
then
\begin{equation}\label{Rineq1}
\frac{dR}{dK}=-\frac{\frac{\partial f}{\partial K}}{\frac{\partial f}{\partial R}}>0.
\end{equation}
Since $f(0)=\frac{b\rho_1}{K\alpha_1}\hat S_1>0$ and $f(\hat S_1)=\frac{b}{K}\hat S_1(\mu_0-\mu_4')<0,$ equation \eqref{G3.12} has a unique solution in $(0,\hat S_1)$ which depends on $K$.This implies that $R(K^*)=0.$
Moreover, for $S$ component, it follows from the first and third equation of  \eqref{eqpt4} that
\begin{equation}\label{G3.14}
I_1=\frac{(\mu_4'-\mu_0)R}{\alpha_1R+\rho_1}.
\end{equation}
Using \eqref{G3.14} in \eqref{sum3.1}, we get
\begin{equation}\label{G3.15}
 S=\frac{\alpha_1(\mu_1-\mu_4')+\frac{\rho_1}{R}(\mu_1-\mu_0)}{\alpha_1(\alpha_1+\frac{\rho_1}{R})}>0.
 \end{equation}
 Thus the equilibrium point $G_3$ exist when $\hat S_1>0$ or  equivalently $S^{**}>\sigma_1.$
 We finish this section by the following property of $R$ which will be helpful in the study of stability of equilibrium point $G_3.$
 \begin{proposition}\label{prop4}
 	Let the equilibrium point $G_3$ exists. Then
	the following inequality holds for $R$,
 $$0<K\frac{dR}{dK}<\sigma_1+R.$$

\end{proposition}
\begin{proof}
	The proof for the  first part of inequality  is given by \eqref{Rineq1}.
	For the second inequality, it follows from \eqref{G3.13} that
	\begin{equation}
	\begin{split}
	K\frac{\partial f}{\partial K}&=\left((\hat S_1-R)+\frac{K}{b}(\mu_0-\mu_4')+\sigma_1+R\right)R+\frac{\rho_1}{\alpha_1}(\hat S_1-R+\sigma_1+R)=0,\\
	K\frac{\partial f}{\partial K}&=f+(\sigma_1+R)(R+\frac{\rho_1}{\alpha_1}).
	\end{split}
	\end{equation}
	Since $f(R,K)=0$ then we have
	\begin{equation}
	\begin{split}
	K\frac{\partial f}{\partial K}&=(\sigma_1+R)(R+\frac{\rho_1}{\alpha_1}).
	\end{split}
	\end{equation}
	Next, it directly follows from \eqref{G3.12.1} that
	\begin{equation}\label{G3.12.2}
	\frac{\partial f}{\partial R}<-R-\frac{\rho_1}{\alpha_1}(\hat S_1-R).
	\end{equation}
	This implies that
	$$-	\frac{\partial f}{\partial R}>R+\frac{\rho_1}{\alpha_1}(\hat S_1-R). $$
	Therefore from \eqref{Rineq1}, we have
	$$K\frac{\partial R}{\partial K}<\sigma_1+R.$$
	
\end{proof}
\subsubsection{Effect of large $K$ on $R$}\label{KG3}

It follows from the first equation of \eqref{eqpt4}
\begin{equation}\label{Ivalue4.1}
I_1=\frac{b}{K\alpha_1}(S^{**}-\sigma_1-R)=\frac{b}{K\alpha_1}(\hat S_1-R).
\end{equation}
Solving the first and second equation of \eqref{eqpt4} together gives
\begin{equation}\label{sum4.1}
S+I_1=\frac{\mu_1-\mu_0}{\alpha_1}.
\end{equation}
Substituting \eqref{Ivalue4.1} and \eqref{sum4.1} in the last equation of \eqref{eqpt4} gives
\begin{equation}
-\frac{b}{K}R^2+\frac{b}{K}\left(R^{**}-\frac{\mu_1-\mu_0}{\alpha_1}-\frac{\rho_1}{\alpha_1}\right)R+\frac{b\rho_1}{K\alpha_1}\left(S^{**}-\frac{\mu_1-\mu_0}{\alpha_1}\right)=0,
\end{equation}
where $$R^{**}=\frac{K}{b}(b-\mu_4').$$

Some simple arguments shows that for large $K$
\begin{equation}\label{largeR}
\begin{split}
R= R^{**}+O(1),\\
I_1=\frac{\mu_4'-\mu_0}{\alpha_1}+O(\frac{1}{K}),\\
S=\frac{\mu_1-\mu_4'}{\alpha_1}+O(\frac{1}{K}).
\end{split}
\end{equation}
So if we increase the carrying capacity then the recovered population grows.
\subsubsection{Dependence of $I_1$ on K}
Let us consider $I_1$ as a function of $K$, then it follows from \eqref{G3.11} that
\begin{equation}
\frac{dI_1}{dK}=\frac{b}{K^2\alpha_1}(R+\sigma_1-KR').
\end{equation}
Note that $\frac{dI_1}{dK}>0$ due to proposition \eqref{prop4}. Also if $S^{**}=\sigma_1,$ then $R=I_1=0.$
Next, using first equation in \eqref{largeR}, we get
\begin{equation}
I_1(K)\rightarrow \frac{b}{K\alpha_1}(S^{**}-R^{**})=\frac{\mu_4'-\mu_0}{\alpha_1}\quad\text{as}\quad K\rightarrow \infty.
\end{equation}
Since the function $I_1$ is monotone, we have
$$I_1(K)<\frac{\mu_4'-\mu_0}{\alpha_1}.$$
\subsubsection{Dependence of $S$ on K}
 Consider \eqref{G3.15} and rewrite it as follows
\begin{equation}\label{G3.16}
 S=\frac{\alpha_1R(\mu_1-\mu_4')+\rho_1(\mu_1-\mu_0)}{\alpha_1(\alpha_1R+\rho_1)},
\end{equation}
and
\begin{equation}
\frac{dS}{dK}=\frac{\rho_1R'(\mu_0-\mu_4')}{(\alpha_1R+\rho_1)^2}.
\end{equation}
It is evident from proposition \ref{prop4}  and \eqref{assum2} that $\frac{dS}{dK}<0$.
Since $R(K^*)=0,$ then from \eqref{G3.16} it follows that
\begin{equation}\label{G3.17}
S(K^*)=\frac{\mu_1-\mu_0}{\alpha_1}.
\end{equation}
Next, using first equation in \eqref{largeR}, we get
\begin{equation}\label{G3.18}
S\rightarrow\frac{\alpha_1(\mu_1-\mu_4')+\frac{\rho_1}{R^{**}}(\mu_1-\mu_0)}{\alpha_1(\alpha_1+\frac{\rho_1}{R^{**}})}=\frac{\mu_1-\mu_4'}{\alpha_1}\quad\text{as}\quad K\rightarrow \infty.
\end{equation}
Relation \eqref{G3.18} and \eqref{G3.17} gives
$$\frac{\mu_1-\mu_4'}{\alpha_1}<S(K)<\frac{\mu_1-\mu_0}{\alpha_1}.$$
\subsubsection{Local stability }
The Jacobian matrix evaluated at $G_3=(S^{*},I_1^*,0,0,R^*)$ is
\begin{gather*}
\setlength{\arraycolsep}{.9\arraycolsep}
\text{\footnotesize$\displaystyle
	J=
	\begin{bmatrix}
	-\frac{b}{K}S^* & -(\alpha_1+\frac{b}{K})S^* & -(\alpha_2+\frac{b}{K})S^* & -(\hat{\alpha_3}+\frac{b}{K})S^* & -\frac{b}{K}S^*  \\
	(\alpha_1-\frac{b}{K})I_1^* & -\frac{b}{K}I_1^*& -(\gamma_1+\frac{b}{K})I_1^* & \beta_1S^*-(\frac{b}{K}+\eta_1)I_1^* & -\frac{b}{K}I_1^* \\
	0 & 0 & b\left(1-\frac{N}{K}\right)+\alpha_2S^*-\gamma_2I_1^*-\mu_2 & \beta_2S^* & 0 \\
	0 & 0 & (\gamma_1+\gamma_2)I_1^* & b\left(1-\frac{N}{K}\right)+\alpha_3S^*+\eta_1I_1^*-\mu_3 & 0\\
	-\frac{b}{K}R^* & \rho_1 -\frac{b}{K}R^* & \rho_2-\frac{b}{K}R^* &\rho_3-\frac{b}{K}R^* & -\rho_1\frac{I_1^*}{R^*}-\frac{b}{K}R^*
	\end{bmatrix}.
	$}
\end{gather*}
Due to the block structure of the matrix the stability of the above matrix is equivalent to the stability of the following  two  matrices
\begin{equation}\label{J1}
J_1=
\begin{bmatrix}
-\frac{b}{K}S^* & -(\alpha_1+\frac{b}{K})S^* &  -\frac{b}{K}S^* \\
(\alpha_1-\frac{b}{K})I_1^* & -\frac{b}{K}I_1^* & -\frac{b}{K}I_1^* \\
-\frac{b}{K}R^* &  \rho_1 -\frac{b}{K}R^* & -\rho_1\frac{I_1^*}{R^*}-\frac{b}{K}R^*
\end{bmatrix},
\end{equation}
and
\begin{equation}\label{J2}
J_2=
\begin{bmatrix}
b(1-\frac{N}{K})+\alpha_2S^*-\gamma_2I_1^*-\mu_2 &  \beta_2S^* \\
(\gamma_1+\gamma_2)I_1^* &  b(1-\frac{N}{K})+\alpha_3S^*+\eta_1I_1^*-\mu_3
\end{bmatrix}.
\end{equation}
The characteristic polynomial of matrix $J_1$ is given by
$$\lambda^3+a_2\lambda^2+a_1\lambda+a_0=0$$
where
\begin{equation}
\begin{split}
a_0&=\rho_1\alpha_1^2\frac{SI_1^2}{R}+\frac{b\alpha_1^2}{K}SI_1R+\frac{b\alpha_1\rho_1}{K}SI_1,\\
a_1&=\alpha_1^2SI_1+\frac{\rho_1b}{K}\frac{I_1^2}{R}+\frac{\rho_1b}{K}I_1+\frac{\rho_1b}{K}\frac{SI_1}{R},\\
a_2&=\frac{b}{K}N+\frac{\rho_1I_1}{R}.
\end{split}
\end{equation}
By the Routh-Hurwitz criterion,  the stability condition is given by
$$a_2a_1-a_0=\frac{b\alpha_1}{K}(\mu_1'-\mu_0)SI_1+(\frac{b}{K}N+\frac{\rho_1I_1}{R})\left(\frac{\rho_1b}{K}\frac{I_1^2}{R}+\frac{\rho_1b}{K}I_1+\frac{\rho_1b}{K}\frac{SI_1}{R}\right)>0. $$
Therefore  the matrix $J_1$ is stable and the stability of matrix $J$ is equivalent to the stability of matrix $J_2.$
So
\begin{equation}
\begin{split}
\det J_2=\left(\frac{b}{K\alpha_1}(\alpha_1-\alpha_2-\gamma_2)(\hat S-R)+\alpha_2(\sigma_1-\sigma_2)\right)\left(\frac{b}{K\alpha_1}(\alpha_1-\alpha_3+\eta_1)(\hat S-R)+\alpha_3(\sigma_1-\sigma_3)\right)-\beta_2(\gamma_1+\gamma_2)S^*I_1^*.
\end{split}
\end{equation}
We rewrite the  $\det J_2$ as follows
\begin{equation}\label{det1.1}
\Delta_1(K)=\det J_2=P_1Q_1-U_1,
\end{equation}
where
\begin{equation} \label{det1.2}
\begin{split}
P_1(K)&=\left(\frac{b}{K\alpha_1}(\alpha_1-\alpha_2-\gamma_2)(\hat S-R)+\alpha_2(\sigma_1-\sigma_2)\right),\\
Q_1(K)&=\left(\frac{b}{K\alpha_1}(\alpha_1-\alpha_3+\eta_1)(\hat S-R)+\alpha_3(\sigma_1-\sigma_3)\right),\\
U_1(K)&=\beta_2(\gamma_1+\gamma_2)S^*I_1^*.
\end{split}
\end{equation}
The sign of the determinant depends on the sign of $P_1 $ and $Q_1. $
First, we consider
$\Delta_1(K^*)=\alpha_2\alpha_3(\sigma_2-\sigma_1)(\sigma_3-\sigma_1)>0$ where $K^*$ is defined by \eqref{Kcrit}. Also
\begin{equation}
\begin{split}
P_1(K^*)=\alpha_2(\sigma_1-\sigma_2)<0\quad \text{and}\quad Q_1(K^*)=\alpha_3(\sigma_1-\sigma_3)<0,
\end{split}
\end{equation}
and $\det(J_2)>0.$ Hence the matrix $J_2$ is stable at $K=K^*.$ Let us introduce $K_0$ as the largest number, $K_0>K^*,$ such that $P_1(K)<0$ and $Q_1(K)<0$, $K^*\leq K<K_0.$ Let $K_1>K^*$ be the first root of $\det J_2(K)=0. $ So $\det J_2(K)>0$ for $K\in [K^*,K_1). $ Clearly, $K_1<K_0. $ Therefore the matrix $J_2$ is stable on $[K^*,K_1)$ and lose its stability at $K_1.$

Next, for further study of stability of matrix $J_2$ on the root $K_1$, we assume that
\begin{equation}\label{trans1}
\delta_1=(\alpha_1-\alpha_2-\gamma_2)>0\quad\text{and}\quad \delta_2=(\alpha_1-\alpha_3+\eta_1)>0
\end{equation}
 and
\begin{equation}\label{stab1.1}
\frac{\mu_4'-\mu_0}{\alpha_1}<\frac{1}{2}\sigma_1,
\end{equation}
then by proposition \eqref{prop4}, we have
\begin{equation}
\frac{dP_1}{dK}=\frac{b\delta_1}{K^2\alpha_1}(R+\sigma_1-KR')>0
\end{equation}
and
\begin{equation}
\frac{dQ_1}{dK}=\frac{b\delta_2}{K^2\alpha_1}(R+\sigma_1-KR')>0.
\end{equation}
Since $P_1$ and $Q_1$ are negative on $[K^*,K_0]$, then
\begin{equation}
(P_1Q_1)'=P_1'Q_1+P_1Q_1'<0.
\end{equation}
Furthermore for $U_1$ we consider
\begin{equation}
\frac{dSI_1}{dK}=(\sigma_1-2I_1)\frac{dI_1}{dK}=\frac{b}{K^2\alpha_1}(\sigma_1-2I_1)(R+\sigma_1-KR').
\end{equation}
This implies that
$\mu_4'-\mu_0<\frac{\mu_1-\mu_0}{2}. $
Therefore  by \eqref{stab1.1}, $\frac{dSI_1}{dK}>0.$
Hence the function $P_1Q_1$ is decreasing with respect to $K$ on  $[K^*,K_0]$ where $P_1,Q_1<0$ together with $SI_1. $ This shows that $\Delta_1(K)$ is decreasing function and $\Delta_1(K^*)>0.$

Furthermore, for large $K$  it follows from the first equation of \eqref{largeR} that
\begin{equation}
P_1= \frac{\delta_1}{\alpha_1}(\mu_4'-\mu_0)+\alpha_2(\sigma_1-\sigma_2)+O(\frac{1}{K}),
\end{equation}
\begin{equation}
Q_1= \frac{\delta_2}{\alpha_1}(\mu_4'-\mu_0)+\alpha_3(\sigma_1-\sigma_3)+O(\frac{1}{K})
\end{equation}
and
\begin{equation}
U_1= \frac{\beta_2(\gamma_1+\gamma_2)}{\alpha_1^2}(\mu_4'-\mu_0)(\mu_1-\mu_4')+O(\frac{1}{K}).
\end{equation}
In this case
$$\Delta_1(\infty)= \biggl(\frac{\delta_1}{\alpha_1}(\mu_4'-\mu_0)+\alpha_2(\sigma_1-\sigma_2)\biggr)\biggl(\frac{\delta_2}{\alpha_1}(\mu_4'-\mu_0)+\alpha_3(\sigma_1-\sigma_3)\biggr)- \frac{\beta_2(\gamma_1+\gamma_2)}{\alpha_1^2}(\mu_4'-\mu_0)(\mu_1-\mu_4').$$
 If $\Delta_1(\infty)>0,$ then
 \begin{equation}\label{transcrit}
 \biggl(\frac{\delta_1}{\alpha_1}(\mu_4'-\mu_0)+\alpha_2(\sigma_1-\sigma_2)\biggr)\biggl(\frac{\delta_2}{\alpha_1}(\mu_4'-\mu_0)+\alpha_3(\sigma_1-\sigma_3)\biggr)> \frac{\beta_2(\gamma_1+\gamma_2)}{\alpha_1^2}(\mu_4'-\mu_0)(\mu_1-\mu_4')
 \end{equation}
   and the matrix $J_2$ is stable for all $K>K^*$.
But if $\Delta_1(\infty)<0$ then matrix $J_2$ is unstable for all $K>K^*$ .

Hence equilibrium point $G_3$ is locally stable for all $K\in [K^*, K_1)$ and loses its stability at $K_1$ and become unstable for all  $K\geq K_1$.
Furthermore, the matrix $J_2$ has two simple eigenvalue for small $\Delta_1(K)$, one of them is separated from zero and negative and the other is small and given by

$$\lambda(K)=\frac{\Delta_1}{P_1(K)+Q_1(K)}+O(\Delta_1^2).$$
Moreover if $\Delta_1(K)=0$ then we have another equilibrium point $G_5.$
Biologically the stability of $G_3$ indicates that strain one can invade in the population alone  for a certain value of carrying capacity, but as K increases,  the invasion and persistence of more infections is possible.

\subsection{Equilibrium point $G_4$}
\subsubsection{Existence}
  It follows from the first and second equation of \eqref{eqpt5} that
\begin{equation}\label{sum4.2}
S+I_2=\frac{\mu_2-\mu_0}{\alpha_2}=\sigma_2.
\end{equation}
Also from the  first  equation of  \eqref{eqpt5} we get
\begin{equation}\label{G4.11}
I_2=\frac{b}{K\alpha_2}(S^{**}-\sigma_2-R)=\frac{b}{K\alpha_2}(\hat S_2-R), \quad \hat S_2=S^{**}-\sigma_2
\end{equation}
The components $R$ and $I_2$ are positive if $\hat S_2>0$ and $0<R<\hat S_2.$
We obtain equation for $R$ by substituting $I_2$ and $S=\sigma_2-I_2$ in the last equation of \eqref{eqpt5} and dividing by $\frac{b}{K}$ gives
\begin{equation}\label{G4.12}
f(R,K)=0,
\end{equation}
where
$$f(R,K)=\left((\hat S_2-R)+\frac{K}{b}(\mu_0-\mu_4')\right)R+\frac{\rho_2}{\alpha_2}(\hat S_2-R).$$
Differentiating  equation \eqref{G4.12}  w.r.t $R$ gives
\begin{equation}\label{4.12.1}
\frac{\partial f}{\partial R}=\left((\hat S_2-R)+\frac{K}{b}(\mu_0-\mu_4')\right)-R-\frac{\rho_2}{\alpha_2}
\end{equation}
and it follows from \eqref{G4.12}  that $\left((\hat S_2-R)+\frac{K}{b}(\mu_0-\mu_4')\right)<0.$ Therefore $\frac{\partial f}{\partial R}<0.$

Moreover
\begin{equation}\label{G4.13}
\frac{\partial f(R,K))}{\partial K}=\left(\frac
{b-\mu_0} {b}+\frac
{\mu_0-\mu_4'} {b}\right)R+\frac{\rho_2}{\alpha_2}(\frac
{b-\mu_0} {b})
\end{equation}
Since
$$\frac{d f(R(K),K)}{dK}=\frac{\partial f}{\partial R}\frac{dR}{dK}+\frac{\partial f}{\partial K}=0,$$
we have
\begin{equation}\label{Rineq2}
\frac{dR}{dK}=-\frac{\frac{\partial f}{\partial K}}{\frac{\partial f}{\partial R}}>0.
\end{equation}
Since $f(0)=\frac{b\rho_2}{K\alpha_2}\hat S_2>0$ and $f(\hat S_2)=\frac{b}{K}\hat S_2(\mu_0-\mu_4')<0,$ therefore  equation \eqref{G4.12} has a unique solution in $(0,\hat S_2)$ which depends on $K$.
Moreover, for $S$ component, it follows from the first and third equation of  \eqref{eqpt5} that
\begin{equation}\label{G4.14}.
I_2=\frac{(\mu_4'-\mu_0)R}{\alpha_2R+\rho_2}.
\end{equation}
Using \eqref{G4.14} in \eqref{sum4.2}, we get
\begin{equation}\label{G4.15}
S=\frac{\alpha_2(\mu_2-\mu_4')+\frac{\rho_2}{R}(\mu_2-\mu_0)}{\alpha_2(\alpha_2+\frac{\rho_2}{R})}>0.
\end{equation}
Hence the equilibrium point $G_4$ exist when $\hat S_2>0$.

 \begin{proposition}\label{prop5}
	Let the equilibrium point $G_4$ exists then
	the following inequality holds for $R$

		 $$0<K\frac{dR}{dK}<\sigma_2+R.$$

\end{proposition}
\begin{proof}
	The proof for the  first inequality  is given in the case of equilibrium points $G_4$ by \eqref{Rineq2}.
	For the second inequality, it follows from \eqref{G4.13} that
	\begin{equation}
	\begin{split}
	K\frac{\partial f}{\partial K}&=\left((\hat S_2-R)+\frac{K}{b}(\mu_0-\mu_4')+\sigma_2+R\right)R+\frac{\rho_2}{\alpha_2}(\hat S_2-R+\sigma_2+R)=0,\\
	K\frac{\partial f}{\partial K}&= f(R,K)+(\sigma_2+R)(R+\frac{\rho_2}{\alpha_2}).
	\end{split}
	\end{equation}
	Since $ f(R,K)=0,$ we have
	\begin{equation}
	\begin{split}
	K\frac{\partial f}{\partial K}&=(\sigma_2+R)(R+\frac{\rho_2}{\alpha_2}).
	\end{split}
	\end{equation}
	Next, it easily follows from \eqref{4.12.1} that
	\begin{equation}\label{G4.12.2}
	\frac{\partial f}{\partial R}<-R-\frac{\rho_2}{\alpha_2}(\hat S_2-R).
	\end{equation}
	This implies that
	$$-	\frac{\partial f}{\partial R}>R+\frac{\rho_2}{\alpha_2}(\hat S_2-R).$$
	Therefore from \eqref{Rineq2}, we have
	$$K\frac{\partial R}{\partial K}<\sigma_2+R.$$
	
\end{proof}
\subsubsection{Effect of large $K$ on recovered class }
Solving \eqref{eqpt5} , we get

\begin{equation}
-\frac{b}{K}R^2+\frac{b}{K}\left(R^{**}-\frac{\mu_2-\mu_0}{\alpha_2}-\frac{\rho_2}{\alpha_2}\right)R+\frac{b\rho_2}{K\alpha_2}\left(S^{**}-\frac{\mu_2-\mu_0}{\alpha_2}\right)=0
\end{equation}
and some simple arguments shows that
\begin{equation}\label{large R2}
\begin{split}
R= R^{**}+O(1),\\
I_2=\frac{\mu_4'-\mu_0}{\alpha_2}+O(\frac{1}{K}),\\
S=\frac{\mu_2-\mu_4'}{\alpha_1}+O(\frac{1}{K}).
\end{split}
\end{equation}
\subsubsection{Dependence of $I_2$ on K}
Let us consider $I_2$ as a function of $K$, then it follows from \eqref{G4.11} that
\begin{equation}
\frac{dI_2}{dK}=\frac{b}{K^2\alpha_2}(R+\sigma_2-KR')
\end{equation}
Note that $\frac{dI_2}{dK}>0$ due to proposition \eqref{prop5}. Also if $S^{**}=\sigma_2$ then $R=I_2=0.$
Next, using first equation in \eqref{largeR}, we get
\begin{equation}
I_2(K)\rightarrow \frac{b}{K\alpha_2}(S^{**}-R^{**})=\frac{\mu_4'-\mu_0}{\alpha_2}\quad\text{as}\quad K\rightarrow \infty.
\end{equation}
Since the function $I_2$ is monotone, we have
$$I_2(K)<\frac{\mu_4'-\mu_0}{\alpha_2}.$$

\subsubsection{Dependence of $S$ on K}
Consider \eqref{G4.15} and rewrite it as follows
\begin{equation}\label{G4.16}
S=\frac{\alpha_2R(\mu_2-\mu_4')+\rho_2(\mu_2-\mu_0)}{\alpha_2(\alpha_2R+\rho_2)},
\end{equation}
and
\begin{equation}
\frac{dS}{dK}=\frac{\rho_2R'(\mu_0-\mu_4')}{(\alpha_2R+\rho_2)^2}.
\end{equation}
It is evident from proposition \ref{prop5}  and \eqref{assum2} that $\frac{dS}{dK}<0$.
Since $R(\hat K^*)=0,$ then from \eqref{G4.16} it follows that
\begin{equation}\label{G4.17}
S(\hat K^*)=\frac{\mu_2-\mu_0}{\alpha_2}.
\end{equation}
Next, using first equation in \eqref{largeR}, we get
\begin{equation}\label{G4.18}
S\rightarrow\frac{\alpha_2(\mu_2-\mu_4')+\frac{\rho_2}{R^{**}}(\mu_2-\mu_0)}{\alpha_2(\alpha_2+\frac{\rho_2}{R^{**}})}=\frac{\mu_2-\mu_4'}{\alpha_2}\quad\text{as}\quad K\rightarrow \infty.
\end{equation}
Relation \eqref{G4.18} and \eqref{G4.17} gives
$$\frac{\mu_2-\mu_4'}{\alpha_2}<S(K)<\frac{\mu_2-\mu_0}{\alpha_2}.$$
 \subsubsection{Local stability analysis }
 The Jacobian matrix evaluated at $G_4=(S^{*},0,I_2^*,0,R^*)$ is
 \begin{gather*}
 \setlength{\arraycolsep}{.9\arraycolsep}
 \text{\footnotesize$\displaystyle
 	J=
 	\begin{bmatrix}
 	-\frac{b}{K}S^* & -(\alpha_1+\frac{b}{K})S^* & -(\alpha_2+\frac{b}{K})S^* & -(\hat{\alpha_3}+\frac{b}{K})S^* & -\frac{b}{K}S^*  \\
 	0 & b\left(1-\frac{N}{K}\right)+\alpha_1S^*-\gamma_1I_2^*-\mu_1 & 0 & \beta_1S^* & 0 \\
 	(\alpha_2-\frac{b}{K})I_2^* & -(\gamma_2+\frac{b}{K})I_2^* &-\frac{b}{K}I_2^* &  \beta_2S^*-(\frac{b}{K}+\eta_2)I_2^* & -\frac{b}{K}I_2^* \\
 	0 & (\gamma_1+\gamma_2)I_2^* & 0 & b\left(1-\frac{N}{K}\right)+\alpha_3S^*+\eta_2I_2^*-\mu_3 & 0\\
 	-\frac{b}{K}R^*  & \rho_1-\frac{b}{K}R^*  & \rho_2-\frac{b}{K}R^* & \rho_3-\frac{b}{K}R^*  & -\rho_2\frac{I_2^*}{R^*}-\frac{b}{K}R^*
 	\end{bmatrix},
 	$}
 \end{gather*}
 where we partitioned  the Jacobian matrix into one  $2\times2$ and one $3\times3$  blocks with
 \begin{equation}\label{J4.1}
 J_3=
 \begin{bmatrix}
 -\frac{b}{K}S^* & -(\alpha_2+\frac{b}{K})S^* &  -\frac{b}{K}S^* \\
 (\alpha_2-\frac{b}{K})I_2^* & -\frac{b}{K}I_2^* & -\frac{b}{K}I_2^* \\
 -\frac{b}{K}R^* & \rho_2 -\frac{b}{K}R^*& -\rho_2\frac{I_2^*}{R^*}-\frac{b}{K}R^*
 \end{bmatrix},
 \end{equation}
 and
 \begin{equation}\label{J4.2}
 J_4=
 \begin{bmatrix}
 b(1-\frac{N}{K})+\alpha_1S^*-\gamma_1I_2^*-\mu_1 &  \beta_1S^* \\
 (\gamma_1+\gamma_2)I_2^* &  b(1-\frac{N}{K})+\alpha_3S^*+\eta_2I_2^*-\mu_3
 \end{bmatrix},
 \end{equation}
 The stability of the  matrix $J$ is equivalent to the stabliliy of the above two matrices.
 Next, the characterstic polynomial of matrix $J_3$ is given by
 $$\lambda^3+a_2\lambda^2+a_1\lambda+a_0=0$$
 where
 \begin{equation}
 \begin{split}
 a_0&=\rho_2\alpha_2^2\frac{SI_2^2}{R}+\frac{b\alpha_2^2}{K}SI_2R+\frac{b\alpha_2\rho_2}{K}SI_2\\
 a_1&=\alpha_2^2SI_2+\frac{\rho_2b}{K}\frac{I_2^2}{R}+\frac{\rho_2b}{K}I_2+\frac{\rho_2b}{K}\frac{SI_2}{R}\\
 a_2&=\frac{b}{K}N+\frac{\rho_2I_2}{R}
 \end{split}
 \end{equation}
 By the Routh-Hurwitz criterion,  the stability condition is given by
 $$a_2a_1-a_0=\frac{b\alpha_2}{K}(\mu_2'-\mu_0)SI_2+(\frac{b}{K}N+\frac{\rho_2I_2}{R})\left(\frac{\rho_2b}{K}\frac{I_2^2}{R}+\frac{b\rho_2}{K}I_2+\frac{b\rho_2}{K}\frac{SI_2}{R}\right)>0$$
 Therefore  the matrix $J_3$ is stable and the stability of matrix $J$ is equivalent to the stability of matrix $J_4$
 So
 \begin{equation}
 \begin{split}
 \det J_4=\left(\frac{b}{K\alpha_2}(\alpha_2-\alpha_1-\gamma_1)(\hat S_2-R)+\alpha_1(\sigma_2-\sigma_1)\right)\left(\frac{b}{K\alpha_2}(\alpha_2-\alpha_3+\eta_2)(\hat S_2-R)+\alpha_3(\sigma_2-\sigma_3)\right)-\beta_1(\gamma_1+\gamma_2)S^*I_2^*.
 \end{split}
 \end{equation}
 We rewrite the  $\det J_4$ as follows
 \begin{equation}\label{det2.1}
 \Delta_2(K)=\det J_4=P_2Q_2-U_2,
 \end{equation}
 where
 \begin{equation} \label{det2.2}
 \begin{split}
 P_2(K)&=\left(\frac{b}{K\alpha_2}(\alpha_2-\alpha_1-\gamma_1)(\hat S_2-R)+\alpha_1(\sigma_2-\sigma_1)\right),\\
 Q_2(K)&=\left(\frac{b}{K\alpha_2}(\alpha_2-\alpha_3+\eta_2)(\hat S_2-R)+\alpha_3(\sigma_2-\sigma_3)\right),\\
 U_2(K)&=\beta_1(\gamma_1+\gamma_2)S^*I_2^*.
 \end{split}
 \end{equation}
 The sign of the determinant depends on the sign of $P_2 $ and $Q_2. $
 First, we consider
 $\Delta_2(\hat K^*)=\alpha_2\alpha_3(\sigma_2-\sigma_1)(\sigma_3-\sigma_1)<0$ where $\hat K^*==\frac{b\sigma_2}{b-\mu_0}.$ and
 \begin{equation}\label{P2Q2}
 \begin{split}
 P_2(\hat K^*)=\alpha_2(\sigma_1-\sigma_2)>0\quad \text{and}\quad Q_2(\hat K^*)=\alpha_3(\sigma_1-\sigma_3)<0,
 \end{split}
 \end{equation}
 and $\det(J_4)<0.$  Therefore the matrix $J$ is unstable at $K=\hat K^*.$

Let us  assume that $ \delta_3=(\alpha_2-\alpha_1-\gamma_1)<0,$ and $  \delta_4=(\alpha_2-\alpha_3+\eta_2)>0$.
 Then by proposition \eqref{prop5}, we have
 \begin{equation}
 \frac{dP_2}{dK}=\frac{b\delta_3}{K^2\alpha_2}(R+\sigma_2-KR')<0,
 \end{equation}
 \begin{equation}
 \frac{dQ_2}{dK}=\frac{b\delta_4}{K^2\alpha_2}(R+\sigma_2-KR')>0
\end{equation}
and
 \begin{equation}
\frac{dU_2}{dK}=\frac{d}{dK}(\sigma_2I_2-I_2^2)=\frac{b\sigma_2}{K^2\alpha_2}(R+\sigma_2-KR')\biggl(1-\frac{2b}{K\alpha_2}(S^{**}-\sigma_2-R)\bigg)
\end{equation}
Moreover, for large $K \longrightarrow \infty$ , it follows from the first equation of \eqref{large R2} that
\begin{equation}
P_2(\infty)= \frac{\delta_3}{\alpha_2}(\mu_4'-\mu_0)+\alpha_1(\sigma_2-\sigma_1),
\end{equation}
\begin{equation}
Q_2(\infty)= \frac{\delta_3}{\alpha_2}(\mu_4'-\mu_0)+\alpha_1(\sigma_2-\sigma_1).
\end{equation}
and
\begin{equation}
U_2(\infty)= \frac{\beta_1(\gamma_1+\gamma_2)}{\alpha_2^2}(\mu_4'-\mu_0)(\mu_2-\mu_4')
\end{equation}
Since we assume that $\alpha_3<\alpha_2<\alpha_1$, then
\begin{equation}\label{P2inf}
P_2(\infty)=(\mu_4'-\mu_1)+\frac{\alpha_1}{\alpha_2}(\mu_2-\mu_4')-\frac{\gamma_1}{\alpha_2}(\mu_4'-\mu_0)>\mu_2-\mu_1-\frac{\gamma_1}{\alpha_2}(\mu_4'-\mu_0)>0
\end{equation}
\begin{equation}
Q_2(\infty)= (\mu_4'-\mu_3)+\frac{\alpha_3}{\alpha_2}(\mu_2-\mu_4')+\frac{\eta_2}{\alpha_2}(\mu_4'-\mu_0)<\mu_2-\mu_3+\frac{\eta_2}{\alpha_2}(\mu_4'-\mu_0)<0
\end{equation}
It follows from the first equation in \eqref{P2Q2} , from \eqref{P2inf} and monotonicity of $P_2$ that $P_2$ is positive for all $K. $ Since $P_2$ is always positive, to guarantee local stability, $Q_2$ must be negative but in that case $\Delta_2(K)<0.$ This shows that the equilibrium point $G_4$ is always unstable.

\subsection{Coexistence Equilibrium  point $G_5$}
 \subsubsection{Local stability analysis }

 The Jacobian matrix evaluated at $G_5=(S^{*},I_1^*,I_2^*,I_{12}^*,R^*)$ is
\begin{gather*}
\setlength{\arraycolsep}{.9\arraycolsep}
\text{\footnotesize$\displaystyle
	J=
	\begin{bmatrix}
	-\frac{b}{K}S^* & -(\alpha_1+\frac{b}{K})S^* & -(\alpha_2+\frac{b}{K})S^* & -(\hat{\alpha_3}+\frac{b}{K})S^* & -\frac{b}{K}S^*  \\
		(\alpha_1-\frac{b}{K})I_1^* +\beta_1I_{12}^*& -\beta_1\frac{S^*I_{12}^*}{I_1^*}-\frac{b}{K}I_1^*& -(\frac{b}{K}+\gamma_1)I_1^* & \beta_1S^*-(\frac{b}{K}+\eta_1)I_1^* & -\frac{b}{K}I_1^* \\
	(\alpha_2-\frac{b}{K})I_2^*+\beta_2I_{12}^*& -(\frac{b}{K}+\gamma_2)I_2^* & -\beta_2\frac{S^*I_{12}^*}{I_2^*}-\frac{b}{K}I_2^* & \beta_2S^*-(\frac{b}{K}+\eta_2)I_2^* & -\frac{b}{K}I_2^* \\
		(\alpha_3-\frac{b}{K})I_{12}^* & (\eta_1-\frac{b}{K})I_{12}^*+(\gamma_1+\gamma_2)I_2^* & (\eta_2-\frac{b}{K})I_{12}^*+(\gamma_1+\gamma_2)I_1^* & -(\gamma_1+\gamma_2)\frac{I_1^*I_2^*}{I_{12}^*}-\frac{b}{K}I_{12}^* & -\frac{b}{K}I_{12}^* \\
	-\frac{b}{K}R^*  & \rho_1-\frac{b}{K}R^*  & \rho_2-\frac{b}{K}R^* & \rho_3-\frac{b}{K}R^*  & b(1-\frac{N}{K})-\mu_4'-\frac{b}{K}R^*
	\end{bmatrix},
	$}
\end{gather*}
and for $\frac{1}{K}=0$, we can write  the matrix $J$ as follows
\begin{gather*}
\setlength{\arraycolsep}{.9\arraycolsep}
\text{\footnotesize$\displaystyle
	J=
	\begin{bmatrix}
	0 & -\alpha_1S^* & -\alpha_2S^* & -\hat{\alpha_3}S^* & 0  \\
	\alpha_1I_1^* +\beta_1I_{12}^*& -\beta_1\frac{S^*I_{12}^*}{I_1^*} & -\gamma_1I_1^* & \beta_1S^*-\eta_1I_1^* & 0 \\
	\alpha_2I_2^*+\beta_2I_{12}^*& -\gamma_2I_2^* & -\beta_2\frac{S^*I_{12}^*}{I_2^*} & \beta_2S^*-\eta_2I_2^* & 0 \\
	\alpha_3I_{12}^* & \eta_1I_{12}^*+(\gamma_1+\gamma_2)I_2^* & \eta_2I_{12}^*+(\gamma_1+\gamma_2)I_1^* & -(\gamma_1+\gamma_2)\frac{I_1^*I_2^*}{I_{12}^*} & 0 \\
\ast & \ast  & \ast & \ast & b-\mu_4'
	\end{bmatrix},
	$}
\end{gather*}
Next, we consider
\begin{gather*}
\setlength{\arraycolsep}{.9\arraycolsep}
\text{\footnotesize$\displaystyle
	J \sim \tilde{J}=
	\begin{bmatrix}
	0 & -\alpha_1S^* & -\alpha_2S^* & -\hat{\alpha_3}S^*   \\
	\alpha_1I_1^* +\beta_1I_{12}^*& -\beta_1\frac{S^*I_{12}^*}{I_1^*} & -\gamma_1I_1^* & \beta_1S^*-\eta_1I_1^*  \\
	\alpha_2I_2^*+\beta_2I_{12}^*& -\gamma_2I_2^* & -\beta_2\frac{S^*I_{12}^*}{I_2^*} & \beta_2S^*-\eta_2I_2^*  \\
	\alpha_3I_{12}^* & \eta_1I_{12}^*+(\gamma_1+\gamma_2)I_2^* & \eta_2I_{12}^*+(\gamma_1+\gamma_2)I_1^* & -(\gamma_1+\gamma_2)\frac{I_1^*I_2^*}{I_{12}^*}
	\end{bmatrix},
	$}
\end{gather*}

\begin{equation*}
\begin{split}
	\det( \tilde{J})=
	-\beta_1\frac{S^*I_{12}^*}{I_1^*}\begin{vmatrix}
	0  & -\alpha_2S^* & -\hat{\alpha_3}S^*   \\
	\alpha_2I_2^*+\beta_2I_{12}^*&  -\beta_2\frac{S^*I_{12}^*}{I_2^*} & \beta_2S^*-\eta_2I_2^*  \\
	\alpha_3I_{12}^* & \eta_2I_{12}^*+(\gamma_1+\gamma_2)I_1^* & -(\gamma_1+\gamma_2)\frac{I_1^*I_2^*}{I_{12}^*}
	\end{vmatrix}\\
	+\gamma_1I_1\begin{vmatrix}
	0 & -\alpha_1S^*  & -\hat{\alpha_3}S^*   \\
	\alpha_2I_2^*+\beta_2I_{12}^*& -\gamma_2I_2^*  & \beta_2S^*-\eta_2I_2^*  \\
	\alpha_3I_{12}^* & \eta_1I_{12}^*+(\gamma_1+\gamma_2)I_2^* & -(\gamma_1+\gamma_2)\frac{I_1^*I_2^*}{I_{12}^*}
	\end{vmatrix}\\
	-\beta_2\frac{S^*I_{12}^*}{I_2^*}\begin{vmatrix}
	0 & -\alpha_1S^*  & -\hat{\alpha_3}S^*   \\
	\alpha_1I_1^* +\beta_1I_{12}^*& -\beta_1\frac{S^*I_{12}^*}{I_1^*}  & \beta_1S^*-\eta_1I_1^*  \\
	\alpha_3I_{12}^* & \eta_1I_{12}^*+(\gamma_1+\gamma_2)I_2^* &  -(\gamma_1+\gamma_2)\frac{I_1^*I_2^*}{I_{12}^*}
	\end{vmatrix}\\
	+\gamma_2I_2\begin{vmatrix}
	0 & -\alpha_2S^*  & -\hat{\alpha_3}S^*   \\
	\alpha_1I_1^*+\beta_1I_{12}^*& -\gamma_1I_1^*  & \beta_1S^*-\eta_1I_1^*  \\
	\alpha_3I_{12}^* & \eta_2I_{12}^*+(\gamma_1+\gamma_2)I_1^* & -(\gamma_1+\gamma_2)\frac{I_1^*I_2^*}{I_{12}^*}
	\end{vmatrix}\\
	+	\begin{vmatrix}
	0 & -\alpha_1S^* & -\alpha_2S^* & -\hat{\alpha_3}S^*   \\
	\alpha_1I_1^* +\beta_1I_{12}^*& 0& 0& \beta_1S^*-\eta_1I_1^*  \\
	\alpha_2I_2^*+\beta_2I_{12}^*&0 & 0& \beta_2S^*-\eta_2I_2^*  \\
	\alpha_3I_{12}^* & \eta_1I_{12}^*+(\gamma_1+\gamma_2)I_2^* & \eta_2I_{12}^*+(\gamma_1+\gamma_2)I_1^* & -(\gamma_1+\gamma_2)\frac{I_1^*I_2^*}{I_{12}^*}
	\end{vmatrix},
\end{split}
\end{equation*}
\begin{equation*}
\begin{split}
=
-\beta_1\frac{S^*I_{12}^*}{I_1^*}\begin{vmatrix}
0  & -\alpha_2S^* & -\hat{\alpha_3}S^*   \\
\alpha_2I_2^*+\beta_2I_{12}^*&  -\beta_2\frac{S^*I_{12}^*}{I_2^*} & \beta_2S^*-\eta_2I_2^*  \\
\alpha_3I_{12}^* & \eta_2I_{12}^*+(\gamma_1+\gamma_2)I_1^* & -(\gamma_1+\gamma_2)\frac{I_1^*I_2^*}{I_{12}^*}
\end{vmatrix}\\
+\gamma_1I_1\begin{vmatrix}
0 & -\alpha_1S^*  & -\hat{\alpha_3}S^*   \\
\alpha_2I_2^*+\beta_2I_{12}^*& -\gamma_2I_2^*  & \beta_2S^*-\eta_2I_2^*  \\
\alpha_3I_{12}^* & \eta_1I_{12}^*+(\gamma_1+\gamma_2)I_2^* & -(\gamma_1+\gamma_2)\frac{I_1^*I_2^*}{I_{12}^*}
\end{vmatrix}\\
-\beta_2\frac{S^*I_{12}^*}{I_2^*}\begin{vmatrix}
0 & -\alpha_1S^*  & -\hat{\alpha_3}S^*   \\
\alpha_1I_1^* +\beta_1I_{12}^*& -\beta_1\frac{S^*I_{12}^*}{I_1^*}  & \beta_1S^*-\eta_1I_1^*  \\
\alpha_3I_{12}^* & \eta_1I_{12}^*+(\gamma_1+\gamma_2)I_2^* &  -(\gamma_1+\gamma_2)\frac{I_1^*I_2^*}{I_{12}^*}
\end{vmatrix}\\
+\gamma_2I_2\begin{vmatrix}
0 & -\alpha_2S^*  & -\hat{\alpha_3}S^*   \\
\alpha_1I_1^*+\beta_1I_{12}^*& -\gamma_1I_1^*  & \beta_1S^*-\eta_1I_1^*  \\
\alpha_3I_{12}^* & \eta_2I_{12}^*+(\gamma_1+\gamma_2)I_1^* & -(\gamma_1+\gamma_2)\frac{I_1^*I_2^*}{I_{12}^*}
\end{vmatrix}\\
+ \Delta
\end{split}
\end{equation*}
where
\begin{eqnarray*}
&&\Delta=\biggl((\beta_1S^*-\eta_1I_1^*)(\alpha_2I_2^*+\beta_2I_{12}^*)-(\beta_2S^*-\eta_2I_2^*)
(\alpha_1I_1^*+\beta_1I_{12}^*)\biggr)\\
&&\times\biggl((\alpha_1\eta_2-\alpha_2\eta_1)S^*I^*_{12}
+(\gamma_1+\gamma_2)S^*(\alpha_1I_1^*-\alpha_2I_2^*)\biggr)
\end{eqnarray*}

Clearly, if all  $\gamma_i, \beta_j$ are zero then $	\det( \tilde{J})=(\alpha_2\eta_1-\alpha_1\eta_2)^2SI_1I_2I_{12}>0.$

\section{Transition dynamics}

As it has been already shown in papers  \cite{SKTW18a}, \cite{SKTW19} that for  any $K>0$ and admissible choice of parameters, there exist exactly one stable equilibrium point which depends continuously on $K$ and those parameters. One can expect the similar transition dynamics in this model as well. Biologically the transition dynamics shows the relationship between transmission parameters and characteristic of corresponding equilibrium state. It also shows that how the transition of equilibrium points depends on carrying capacity which means the invasion of of pathogens increases with growing $K$.
In this study we have the following two possible scenarios under the assumptions \eqref{trans1}
\begin{itemize}
	\item[(i)]  $G_2\longrightarrow G_3$
 	\item[(ii)] $G_2\longrightarrow G_3 \longrightarrow G_5$
 \end{itemize}
The first case exist when \eqref{transcrit} holds.
In that case small $K$ there is a stable disease free equilibrium point but after a threshold, when $K$ continue to increase, that equilibrium point become unstable and the other equilibrium point $G_3$ which corresponds to the existence of only one infection becomes stable. In next case we start with stable disease free equilibrium point for small $K$ and as K increases it bifurcates  from $G_3$ point to an inner point  $G_5$.
This biologically means that after a certain threshold value of carrying capacity disease can invade and  persist in population. This continuity of transition dynamics will hold true for small $ \beta_i$ and $\gamma_j.$

\section{Global stability analysis and some basic properties }
\subsection{Basic reproduction ratio for disease free equilibrium point $G_2=(S^{**},0,0,0,0)$}
The basic reproduction number is the key parameter in epidemiology. It shows the number of secondary cases which one case would produce in a  susceptible population. If $R_0 > 1$, the disease can persist in the population and if $R_0 < 1, $ it will go extinct.
Therefore it is important to estimate  the value of $R_0 $ in a particular population.

In our first model the basic reproductive number is largest value among the following three. i.e  $$ R_0=\max \{\lambda_1,\lambda_2,\lambda_3\}=\max\{R_{01},R_{02},R_{03}\} $$
where
\begin{equation}
\begin{split}
R_{01}=\frac{S^{**}}{b-\mu_1}(\frac{b}{K}-\alpha_1),\\
R_{02}=\frac{S^{**}}{b-\mu_2}(\frac{b}{K}-\alpha_2),\\
R_{03}=\frac{S^{**}}{b-\mu_3}(\frac{b}{K}-\alpha_3).
\end{split}
\end{equation}

\subsection{Global stability analysis for $G_2$}
In this section we formulate a global stability result concerning the disease free equilibrium point which guarantees that the disease can not invade and go extinct in small populations.
\begin{proposition}\label{G.SG2}
Let \eqref{assum2} holds and
$$S^{**}< \min\left({\frac{\mu_1'-\mu_0}{\alpha_1},\frac{\mu_2'-\mu_0}{\alpha_2},\frac{\mu_3'-\mu_0}{\hat \alpha_3}}\right)$$
	
	then the equilibrium point $G_2$ is globally asymptotically stable i.e
	\begin{align*}
	\lim_{t\to\infty}I_1(t)&=
	\lim_{t\to\infty}I_2(t)=
	\lim_{t\to\infty}I_{12}(t)=
	\lim_{t\to\infty}R(t)=0,\\
	\lim_{t\to\infty} S(t)&= S^{**},
	\end{align*}
\end{proposition}
\begin{proof}
	We define a Lyapunov function $v$ such that
	\begin{equation}\label{veq1.1}
	v(t):=S- S^{**}\ln S+I_1+I_2+I_{12}+R.
	\end{equation}
	\begin{equation}
	v(t)':=\frac{(S- S^{**})}{S}S'+I_1'+I_2'+I_{12}'+R',
	\end{equation}
		Differentiating the above equation  and combining with \eqref{Model} yields
	\begin{eqnarray}\label{veq1.2}
&&	v'(t)=-\frac{b}{K}(N-S^{**})^2-\left(\alpha_1S^{**} -(\mu_1'-\mu_0)\right)I_1-\left(\alpha_2S^{**}-(\mu_2'-\mu_0)\right) I_2 \nonumber\\ &&-\left(\hat\alpha_3S^{**}-(\mu_3'-\mu_0)\right)I_{12}-(\mu_4'-\mu_0)R.
	\end{eqnarray}
 Then using \eqref{assum2}, we have
	$v'(t)<0$, i.e. $v(t)$ is an decreasing function  Since $S$, $I_1$, $I_2$, $I_{12},R$ are bounded positive functions, $v(t)$ is bounded too.
	Therefore by Lemma~2 in \cite{SKTW18a}, $N=S$ converges to $S^*$ and $I_1,I_2$, $I_{12}$ converge to zero.	
\end{proof}

\section{Discussion}

In this paper we discussed four model with coinfection to analyse the effect of four different phenomenon. In the first model we considered an SIR model with coinfection of two disease or strains,. We assumed that the growth of each class is regulated by the carrying capacity. Traditionally, the term carrying capacity have  meaning as food resource since food can be the limiting factor for population growth but it can not be necessarily be the food in the case of disease there could be numerous other factors which can limit the population growth (exposure to disease etc.). The density dependence in each class plays an important role in disease dynamics. First in this case we have five equilibrium points depends on $K$. We analysed all equilibrium points and have the local stability analysis. An important scenario arises when one equilibrium point loses its stability after a certain threshold value of K and bifurcates into another point. This is the bifurcation from  disease free state to the states where  at least one of the infected class is non zero. Contrary to \cite{SKTW18a,SKTW19}, we considered the recovered class as well to see the effect of carrying capacity on recovered population. Although including the density dependence factor usually provide sufficient control of population growth and to keep population size bounded. But here we  have observed that the recovered population is not uniformly bounded with respect to K. As K approaches to infinity, recovered class also goes to infinity. So if disease persist in the population then for large K then relaxing the assumption of recovered class is completely immune will gives more interesting results.
Moreover there are three stable equilibrium points for different values of $K$. The dynamics here corresponds to the transition dynamics in \cite{SKTW18a,SKTW19}, which means as $K$ increases, we move form one equilibrium point to another one. It is observed that, similar to previous cases, the dynamics depends on K. When K is small, the disease can not invade in the population and go extinct quickly due to the small size of population.

\end{document}